\documentclass[12pt]{article}

\include{mytexmacros}
\usepackage{natbib}
\usepackage{amssymb,float,subfig}
\usepackage{amsmath}
\usepackage{amsthm}
\usepackage{slashbox}
\usepackage{graphicx}
\usepackage{epsfig}
\include{mytexmacros}
\usepackage{caption}
\usepackage{blindtext}

\usepackage{graphicx}
\usepackage{amsmath,amsthm,bm,mathrsfs}

\newcommand{\pt}[2]{\frac{\partial{#1}}{\partial{#2}}}

\newtheorem{theorem}{Theorem}

\title{Intimal Growth in Cylindrical Arteries: Impact of Anisotropic Growth on Glagov Remodeling}
\date{}

\author{Navid Mohammad Mirzaei, \\
	Dept. of Mathematical Sciences, \\
	University of Delaware, \\
	Newark,\\
	DE 19716 \and
	Pak-Wing Fok, \\
	Dept. of Mathematical Sciences, \\
	University of Delaware,\\
	Newark,\\
	DE 19716}

\begin{document}

\title{Intimal Growth in Cylindrical Arteries: Impact of Anisotropic Growth on Glagov Remodeling}
\author{ {\sc Navid Mohammad Mirzaei}\\[2pt]
Dept. of Mathematical Sciences, University of Delaware, \\
 Newark, DE, 19716, USA.\\[6pt]
 {\sc Pak-Wing Fok}\\[2pt]
 Dept. of Mathematical Sciences, University of Delaware, \\
 Newark, DE, 19716, USA.\\[6pt]
{\rm }\vspace*{6pt}}
\pagestyle{headings}
\markboth{N. M.Mirzaei}{\rm IMPACT OF ANISOTROPIC GROWTH ON GLAGOV REMODELING}
\maketitle


\begin{abstract}
{In this paper we investigate the effect of anisotropic growth on Glagov remodeling in different cases: pure radial, pure circumferential, pure axial and general anisotropic growth. We use the theory of morphoelasticity on an axisymmetric arterial domain. For each case we explore their specific effect on the Glagov curves and stress and  provide the changes in collagen fibers angles in the intima, media and adventitia. In addition, we compare the strain energy produced by growth in radial, circumferential and axial direction and deduce that anisotropic growth generally leads to lower strain energy than isotropic growth. Therefore, we explore an anisotropic growth regime and use the resulting model to simulate vessel remodeling. We compare the Glagov curves, stress, energies and fiber angles in the anisotropic case with those of the isotropic case. Our results show that the anisotropic growth produces a remodeling curve more consistent with Glagov's experimental data with gentler outward remodeling and more realistic stress profiles. }
{Glagov remodeling, morphoelasticity, anisotropic growth, arterial biomechanics, atherosclerosis, intimal thickening.}
\end{abstract}

\section{Introduction}


Despite the recent advances in preventive methods, cardiovascular disease is still the global leading cause of death. According to the American Heart Association, it accounts for more than 17.9 million deaths per year in 2015 and is expected to grow to more than 23.6 million by 2030.

{\bf Atherosclerosis} is a cardiovascular disease which causes the narrowing of the blood vessels  therefore reducing the blood flow. It can lead to life-threatening problems including heart attack and stroke.

According to \cite{Virmani2008}, the evolution of vascular disease involves a combination of endothelial dysfunction, extensive lipid deposition in the intima, exacerbated immune responses, proliferation of vascular smooth muscle cells and remodeling, resulting in the formation of an atherosclerotic plaque. High risk atherosclerotic plaques (vulnerable plaques) have a large lipid-rich necrotic core with an overlying thin fibrous cap infiltrated by inflammatory cells and diffuse calcification. These plaques are more susceptible to rupture. About $4-13\%$ of fatal cases of acute myocardial infarction are caused by rupture (see \cite{London1965}).

Low shear stress plays an essential role in triggering atherosclerosis (see \cite{Libby2002,Mundi2017,Channon2006}). Healthy endothelial cells produce a certain amount of Nitric Oxide (NO) which is a vasodilator. Decrease in laminar shear stress reduces the production of this chemical which leads to endothelial dysfunction. This increases the permeability of the endothelium to low density lipoproteins (LDL) as well as the production of vascular cell adhesion molecule-1 (VCAM-1). These molecules start an inflammatory process by binding with the intracellular adhesion molecule-1 (ICAM-1) on the surface of leukocytes present in the blood stream. Attached to the endothelium, these leukocytes penetrate the vessel wall in response to the chemoattractant (MCP-1) present in the intima. Once inside, macrophage colony stimulating factor (M-CSF) causes them to turn into macrophages. LDLs absorbed by the intima go through oxidization and turn into oxidized LDLs. Macrophages cause inflammation and consume these oxidized LDLs and release more (MCP-1), turning into foam cells. Smooth muscle cells (SMCs) can also migrate into the plaque from the underlying media. The death of SMCs, foam cells and macrophages all contribute to a necrotic core, one of the defining characteristics of a vulnerable plaque (see \cite{Libby2002,Virmani2008}). 

Intima thickness is one of the most important factors in assessing cardiovascular risk and one of the most common methods in measuring the progression of atherosclerosis. Intimal thickening however is different from atherosclerosis since its associated lesions are less inflamed. In other words inflammatory agents like macrophages are almost non existent inside the intima. However, intimal thickening is considered to be an important precursor to atherosclerosis. A thickened intima provides a great opportunity for the onset of atherosclerotic lesions (see \cite{Kim1985, Schwartz1995}).    

Investigating the mechanical properties of the blood vessels is an essential step towards understanding cardiovascular diseases. JD Humphrey and LA Taber investigate the stress-modulated growth and residual stress of the arteries using the concept of opening angles (see \cite{Taber2001}). In their research they speculate that the vascular heterogeneity must be a result of collagen distribution. Four years later, Gerhard Holzapfel {\it et al.} determine the mechanical properties of coronary artery layers with nonatherosclerotic intimal thickening (see \cite{Holzapfel2005}). In their study, they experiment on thirteen hearts, from 3 women and 10 men which were harvested within 24 hours of death. Then they create coronary artery cross sections and cut them along the axial direction to obtain flat rectangular sheets. Thereafter, by exposing the sheets to tensile stresses, they were able to come up with layer-specific mechanical parameters later used in their strain energy function. This function is able to capture the stiffening effect of collagen fibers that exist in each layer.

There are many studies that try to understand the cell and chemical dynamics of intimal thickening and atherosclerosis using reaction-diffusion type models. One can find a comprehensive example of such in Hao and Friedman's study (see \cite{Hao2014}). They have most of the key players including a velocity field which is the result of movement of macrophages, T-cells and smooth muscle cells into the intima. This procedure promotes intimal thickening. Their model however, does not consider the mechanical properties of the intima and neglects the other two layers of the vessel wall. For this reason it qualifies as a reaction-diffusion type model. Mary R. Myerscough {\it et al.} use differential equations to purely explore the dynamics of early atherosclerosis (see \cite{Mary2015}). Their model considers the concentration of LDLs, chemoattractants, embryonic stem (ES) cytokines, macrophages and foam cells. All of their simulations are done in one dimension and their result provides qualitative and quantitative insight into the effect of LDL penetration in the inflammatory response. In 2017, Mary R. Myerscough {\it et al.} further investigate the effect of High density Lipoproteins (HDL) in plaque regression (see \cite{Mary2017}). El Khatib {\it et al.} suggest that inflammation propagates in the intima as a reaction diffusion wave (see \cite{Khatib2007}). They conclude that in the case of intermediate LDL concentrations there are two stable equilibria: one corresponding to the disease free state and the other one to the inflammatory state while the traveling wave connects these two states.  

In 1987, Seymour Glagov discovered an important behavior of the arteries experimenting on section of the left main coronary artery in 136 hearts obtained at autopsy (see \cite{Glagov1987}). He found that arteries remodel as the plaque grows to compensate for the narrowing of the lumen to maintain the histological blood flow. However, this compensation will continue until the lesion occupies about 40\% of the internal elastic lamina area and then the narrowing of the lumen starts. Understanding this phenomenon is of great importance. Since the coronary angiography can only visualize the lumen the extent of the plaque burden in the arterial wall might be underestimated during the compensation phase. Therefore, understanding this attribute of the blood vessels is crucial for devising new methods for determining the severity of arterial diseases such as atherosclerosis. This phenomenon has been the subject of biological and mathematical studies ever since (see \cite{Korshunov2004, Mohiaddin2004, Korshunov2007} and \cite{Fok2016}). PW Fok explores the growth in a 2D annulus subject to a uniform isotropic growth tensor (see \cite{Fok2016}). Although these assumptions are not realistic but the results seem to follow the general attribute of the Glagov remodeling. In other words, it captures the compensation phase followed by an inward remodeling of the endothilial wall at about 30\% stenosis. 

In this paper we focus on a three dimensional axisymmetric vessel wall with 3 layers. We use a finite element method based on morphoelasticity. We utilize the layer specific strain energy function proposed in \cite{Holzapfel2005} to account for the stiffening effect of the collagen fibers. All of our numerical simulations are carried out in a FEniCS framework (see \cite{Fenics2016}). Although there are various studies involving the artery growth in two dimensions, we believe that growth in 3 dimensions produces interesting results that should not be neglected. We provide results that show the isotropic growth assumption is not energetically favorable and produces results that are not realistic such as quick outward remodeling with respect to stenosis. On the other hand, anisotropic treatment of the problem is more reasonable. It results in a gentle outward remodeling with respect to stenosis that is more in line with what Glagov observes (see \cite{Glagov1987}).

This paper is laid out in the following way. In section 2 we discuss the hyperelastic modeling of our problem. In  section 3 we provide our results and finally we summarize our conclusions in section 4.

\section{The variational formulation}
Morphoelasticity is the underlying assumption for our simulations. It interprets the deformation in hyperelastic materials as a pure growth accompanied by an elastic response (see \cite{Goriely2007, Rodriguez1994}). In other words, we can decompose the deformation gradient into a growth tensor ${\bf G}$ and an elastic tensor ${\bf F}_e$:
\begin{equation}
{\bf F} = {\bf F}_e {\bf G}.	\label{eq2.1}
\end{equation}
As mentioned before we consider the artery as a three layered growing domain. This growth is a volumetric growth that occurs only inside the intima and can also be accompanied by surface loads. Corresponding to each of the tensors in (\ref{eq2.1}) we have
\begin{eqnarray}
J &=& \mathrm{det}({\bf F}), \label{eq2.2}\\
J_e &=& \mathrm{det}({\bf F}_e), \label{eq2.3}\\
J_g &=& \mathrm{det}({\bf G}). \label{eq2.4}
\end{eqnarray}

Deformation and growth of the artery lead to a change in the strain energy $W$. This strain energy is the sum of energy stored due to the volumetric changes (${\bf \Psi}_{\rm vol}$) and the anisotropic responses (${\bf \Psi}_{\rm aniso}$) of the layer (see \cite{Holzapfel2002}):
\begin{equation*}
W_i = {\bf \Psi}_{\rm vol}^i+{\bf \Psi}_{\rm aniso}^i
\end{equation*}
\begin{equation}
{\bf \Psi}_{\rm vol}^i= \frac{\mu_i}{2} (I_1-3)+\frac{\nu}{1-2\nu} \mu_i (J_e - 1)^2-\mu_i \ln J_e \label{eq2.5}
\end{equation}
\begin{equation}
{\bf \Psi}_{\rm aniso}^i=  \frac{\eta_i}{\beta_i} \left\{ e^{\beta_i \left[ \rho_i (I_4 -1 )_+^2+(1-\rho_i)(I_1 -3)^2 \right]} -1 \right \} \label{eq2.6}
\end{equation}
where $i=1,2,3$ corresponds to intima, media and adventitia; $\mu_i$, $ \eta_i$ are stress-like parameters; and $\beta_i$, $\rho_i$ are dimensionless. Due to the high content of water in each layer we consider them as nearly incompressible materials and therefore we take the Poisson ratio $\nu$ to be close to $0.5$ in all the layers. Also
\begin{align}
I_1 &= \textrm{Tr}({\bf C}_e) = \textrm{Tr}({\bf F}_e^T {\bf F}_e) \label{eq2.7}\\
I_4 &= {\bf b}(R,Z)^T{\bf C}_e {\bf b}(R,Z) \label{eq2.8}
\end{align}
where ${\bf C}_e = {\bf F}_e^T {\bf F}_e$ is the right Cauchy-Green tensor and $I_1$ is its first invariant. To incorporate the direction for which the collagen fibers are aligned in each of the layers we use ${\bf b}(R,Z)$ which is a unit vector. The role of the collagen fibers is included in $I_4$ which will be triggered only if $I_4 >1 $ because the collagen fibers only contribute to the energy when they are stretched not compressed:
\[(I_4 -1 )_+^2 = \begin{cases}
(I_4 -1 )^2 \qquad &\text{if} \ I_4 >1\\
0 &\text{if} \ I_4 \leq 1
\end{cases}\]

We are interested in finding a solution to the following boundary value problem
\begin{eqnarray}
&&\nabla \cdot \boldsymbol{\sigma} = \hat{\bf f}, \hspace{0.4in} \text{on $\omega$} \label{eq2.9}\\
&&\boldsymbol{\sigma}{\bf n} = -p {\bf n}, \qquad \text{on $\partial \omega^{(1)}_1$} \label{eq2.10}\\
&&\boldsymbol{\sigma}{\bf n}= 0,  \hspace{0.56in} \text{on $\partial \omega^{(2)}_3$} \label{eq2.11} \\
&&\boldsymbol{\sigma}{\bf n}|_{\partial \omega^{(2)}_1}+ \boldsymbol{\sigma}{\bf n}|_{\partial \omega^{(1)}_2}=0 \label{eq2.12}\\
&&\boldsymbol{\sigma}{\bf n}|_{\partial \omega^{(2)}_2} + \boldsymbol{\sigma}{\bf n}|_{\partial \omega^{(1)}_3}=0 \label{eq2.13} 
\end{eqnarray}
Where the tensor ${\boldsymbol \sigma}$ is the Cauchy stress tensor, $\hat{\bf f}$ is the body force, $\omega=\bigcup\limits_{i=1}^3 \omega_i$ for $i=1,2,3$ is the three layered domain after deformation and $\partial \omega^{(1)}_i$ is the inner boundary and $\partial \omega^{(2)}_i$  is the outer boundary of the $i$-th layer after the deformation. We consider $p$ to be the only boundary load which in our case is the blood pressure. We denote the outward unit normal vector to the deformed boundary by ${\bf n}$. Also assuming that the deformed arterial segment in our problem has a finite length we add two traction free boundary conditions for the end surfaces 
\begin{eqnarray}
&&\boldsymbol{\sigma}{\bf n}= 0, \hspace{0.75in} \text{on $\partial \omega_L$} \label{eq2.13.L}\\
&&\boldsymbol{\sigma}{\bf n}= 0, \hspace{0.75in} \text{on $\partial \omega_R$} \label{eq2.13.R}
\end{eqnarray}

Even though, (\ref{eq2.9})-(\ref{eq2.13.R}) seem like a typical boundary value problem, due to the convenience of working with the reference domain, $\Omega$ we prefer to use a system that utilizes $\partial \Omega$ for its boundary condition rather than $\partial \omega$, see Figure \ref{fig1ab}(a). Therefore, by applying Nanson's pull back formula and using the first Piola-Kirchoff stress tensor, (\ref{eq2.9})-(\ref{eq2.13.R}) turn into
\begin{eqnarray}
&&\nabla \cdot {\bf T} = \bf{f}, \hspace{0.9in} \text{on $\Omega_1$,  $\Omega_2$, $\Omega_3$ } \label{eq2.14}\\
&&{\bf T N} = -pJ {{\bf F}^{-T}} {\bf N}, \qquad \text{on $\partial \Omega_1^{(1)}$} \label{eq2.15}\\
&&{\bf T N} = 0, \hspace{1.01in} \text{on $\partial \Omega_3^{(2)}$} \label{eq2.16}\\
&& {\bf T N}|_{\partial \Omega_1^{(2)}} + {\bf T N}|_{\partial \Omega_2^{(1)}}=0 \label{eq2.17}\\
&& {\bf T N}|_{\partial \Omega_2^{(2)}} + {\bf T N}|_{\partial \Omega_3^{(1)}}=0 \label{eq2.18}\\
&&{\bf T N} = 0, \hspace{0.75in} \text{on $\partial \Omega_L$} \label{eq2.18.L}\\
&&{\bf T N} = 0, \hspace{0.75in} \text{on $\partial \Omega_R$} \label{eq2.18.R}
\end{eqnarray}
Where the first Piola-Kirchoff stress is
\begin{equation}
{\bf T}= J_g \frac{\partial W}{\partial {\bf F}_e} {\bf G}^{-T} \label{eq2.19}
\end{equation}
(see \cite{Amar2005,Yin2019}). Also ${\bf f}(X,Y,Z) = J \hat{\bf f}$, (see \cite{Gurtin1981}) and ${\bf N}$ is the outward unit normal vector to the reference boundary, see Figure \ref{fig1ab}(a). For solving this problem we use a weak form that is equivalent to (\ref{eq2.14})-(\ref{eq2.18.R}). 

\begin{figure}[h]
	\begin{minipage}{\textwidth}
		\subfloat[]{\includegraphics[scale=0.7]{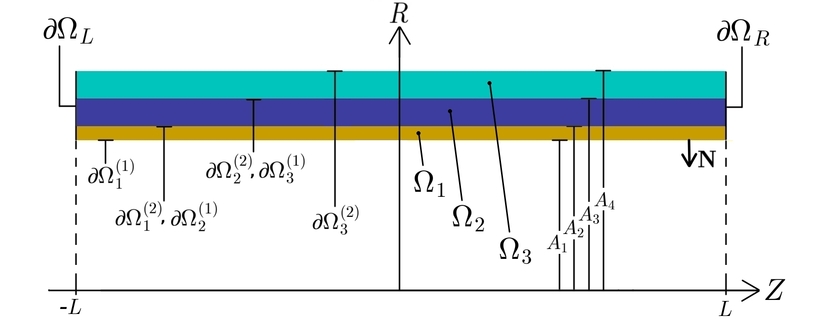}} \\
		\hspace*{1.29in}\subfloat[]{\includegraphics[scale=0.7]{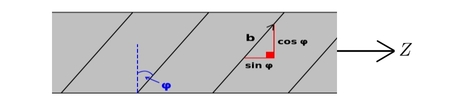}}
	\end{minipage}
	\caption{\footnotesize{(a) Mathematical domain with subdomain and boundary labels schematic. (b) Plan view of the artery and the orientation of a representative fiber. Vector ${\bf b}$ is defined for each layer by (\ref{eq2.1.23}) and the values of $\varphi$ for each layer is given in Table 1.}}\label{fig1ab}
\end{figure} 

\subsection{Weak formulation in cylindrical coordinates}

We consider an axisymmetric cylindrical domain $\Omega$ with three subdomains $\Omega_i$ for $i=1,2,3$ to represent the artery. Let $R$, $\Theta$ and $Z$ be the radius, polar angle and height of a point in the reference domain in cylindrical coordinates and $r$,$\theta$ and $z$ be that of the deformed domain. We consider $(R, \Theta , Z)$ as a generic point in the reference domain and $(r,\theta,z)$ as the one in the deformed domain. Suppose ${\bf u}$ is the displacement field that maps the reference domain into the deformed domain. Then ${\bf T}$ is related to ${\bf u}$ via the following definition for the deformation gradient
\begin{equation}
{\bf F} = {\bf I}+ \nabla {\bf u} \label {eq2.1.20}
\end{equation}
Then the deformation gradient (\ref{eq2.1.20}) in cylindrical coordinates will be given by
\begin{eqnarray}
{\bf F}=
\begin{bmatrix}
\pt{r}{R} & \frac{1}{R}\pt{r}{\Theta} & \pt{r}{Z} \\
r \pt{\theta}{R} & \frac{r}{R} \pt{\theta}{\Theta} & r \pt{\theta}{Z}\\
\pt{z}{R} & \pt{z}{\Theta} & \pt{z}{Z}
\end{bmatrix} \label{eq2.1.21}
\end{eqnarray}
However, in the axisymmetric case $r$ and $z$ are independent of $\Theta$ and $\theta$ is independent of $R$ and $Z$, which simplifies the deformation gradient into
\begin{eqnarray}
{\bf F}=
\begin{bmatrix}
\pt{r}{R} & 0 & \pt{r}{Z} \\
0 & \frac{r}{R}  &  0\\
\pt{z}{R} & 0 & \pt{z}{Z}
\end{bmatrix} \label{eq2.1.22}
\end{eqnarray}
hence
\begin{eqnarray}
J = {\rm det}({\rm F}) =\frac{r}{R} \left(\pt{r}{R} \pt{z}{Z}-\pt{r}{Z}\pt{z}{R}\right) \label{eq2.1.22.J}
\end{eqnarray}
The fiber direction vectors in each layer take the form
\begin{equation}
{\bf b}_i(R,Z) = \cos(\varphi_i) \hat{\bf e}_{\Theta} + \sin(\varphi_i) \hat{\bf e}_Z \label{eq2.1.23}
\end{equation}
where $i=1,2,3$ corresponds to intima, media and adventitia and $\varphi_i$ is the angle from Figure \ref{fig1ab}(b) for each layer. Also $\hat{\bf e}_{\Theta}$ and $\hat{\bf e}_Z$ are the circumferential and axial basis vectors.

Furthermore, ${\bf T}_i= J_{g_i} \pt{W_i}{{\bf F}_{e_i}} {\bf G}_i^{-T}$ for $i=1,2,3$. Using (\ref{eq2.5}), (\ref{eq2.6}), (\ref{eq2.7}), (\ref{eq2.8}), (\ref{eq2.1.22.J}) and (\ref{eq2.1.23}) we have:
\begin{eqnarray}
\pt{W_i}{{\bf F}_{e_i}} &=& \mu_i {\bf F}_e + \frac{2 \mu_i \nu (J_e-1)J_e}{1-2\nu}{\bf F}_e^{-1}- \mu_i {\bf F}_e^{-1} \nonumber\\ &+& \left\{2 \eta_i \rho_i {\bf F}_e {\bf b}_i {\bf b}_i^T (I_4-1)_+ +4 \eta_i (1-\rho_i){\bf F}_e(I_1-3)\right\} e^{\beta_i \left[ \rho_i (I_4 -1 )_+^2+(1-\rho_i)(I_1 -3)^2 \right]}  \label{eq2.1.23.derivative}
\end{eqnarray}
As mentioned before the biology of our problem suggests that the growth occurs only inside the intima. Therefore, ${\bf G}_i = {\bf I}$ and $J_{g_i} =1 $ when $i=2,3$. On the other hand we consider ${\bf G}_1 = \textrm{diag}(g_{\alpha}(Z),g_{\beta}(Z),g_{\gamma}(Z))$ with

\begin{eqnarray}
g_{\alpha}(Z) &=& 1+ \alpha t \exp(-aZ^2) \label{eq2.1.24}\\
g_{\beta}(Z) &=& 1+ \beta t \exp(-aZ^2) \label{eq2.1.25}\\
g_{\gamma}(Z) &=& 1+ \gamma t \exp(-aZ^2) \label{eq2.1.26}
\end{eqnarray}
corresponding to radial, circumferential and axial growth respectively. The variable $t$ is time which is in years throughout this paper. We include the exponential functions in $Z$ to model the effect of local growth in the axial direction. Furthermore, we want growth to increase linearly in time but at different rates and this is the reason for including  $\alpha$,$\beta$ and $\gamma$. In other words, these parameters $\alpha$,$\beta$ and $\gamma$ allow us to explore the effect of anisotropic growth on Glagov remodeling and in the case of isotropic growth we will have $\alpha=\beta=\gamma$. The parameter $a$ determines the locality of growth. We denote the radii of the boundaries between the lumen, intima, media, adventitia and the external tissue in the reference domain by $A_1, A_2, A_3$ and $A_4$ respectively. Also the value $L$ specifies the half-length of the artery cross section such that $-L<Z<L$. See Table 1.

We are now ready to propose a weak form for (\ref{eq2.14})-(\ref{eq2.18.R}).
\begin{theorem}
	Suppose a smooth pressure load $p$ is applied to the inner boundary $\partial \Omega_1^{(1)}$ of a three layered arterial domain $\Omega= \bigcup\limits_{i=1}^3 \Omega_i$ with piecewise smooth boundaries. For simplicity we  denote the outward unit normal vectors ${\bf N}|_{\partial \Omega_i^{(k)}}$, ${\bf N}|_{\partial \Omega_L}$ and ${\bf N}|_{\partial \Omega_R}$  by ${\bf N}_i^{(k)}$, ${\bf N}_L$ and ${\bf N}_R$  for $i=1,2,3$ and $k=1,2$, respectively. Assume that the domain has a finite length $2L$ and is traction free at both ends and  ${\bf f} \in \mathit{L}^2(\Omega)$ and ${\bf G}_i$ for $i=1,2,3$ are growth tensors defined on the intima, media and adventitia respectively. Then defining $J_{g_i} = {\rm det}({\bf G}_i)$ the displacement field ${\bf u} \in \mathit{C}^{2}(\Omega)$ that solves (\ref{eq2.14})-(\ref{eq2.18.R}) also satisfies
	\begin{align}
	2\pi \sum_{i=1}^{3}\int_{-L}^{L} \int_{A_i}^{A_{i+1}}\left[ \left(J_{g_i} \pt{W_i}{{\bf F}_{e_i}} {\bf G}_i^{-T} : \nabla {\bf v}\right) + {\bf f} \cdot{\bf v}  \right]\ R \ dR \ dZ \nonumber \\+ 2\pi \left. \left(\int_{-L}^{L} p J {\bf F}^{-T} {\bf N}_1^{(1)} \cdot {\bf v} \ R \ dZ \right) \right|_{R=A_1}=0  \label{eq2.1.27}
	\end{align}
	for every ${\bf v} \in \mathit{C}^{\infty}(\Omega)$. Where $\pt{W_i}{{\bf F}_{e_i}}$ is defined in (\ref{eq2.1.23.derivative}), $J$ is defined in (\ref{eq2.1.22.J}) and ${\bf F}$ is defined by (\ref{eq2.1.20}) and (\ref{eq2.1.22}).
\end{theorem}
\begin{proof}
	
	Let ${\bf v} \in \mathit{C}^{\infty}(\Omega)$ be arbitrary. By multiplying both sides of (\ref{eq2.14})-(\ref{eq2.18.R}) by $\bf v$ and integrating over their respective domains we get
	\begin{align}
	\int_{\Omega} (\nabla \cdot {\bf T}) \cdot {\bf v} {\bf dx} &= \int_{\Omega}\bf{f} \cdot {\bf v} {\bf dx}, \tag{I}\\
	\int_{\partial \Omega_1^{(1)}} {\bf T}{\bf N}_1^{(1)} \cdot {\bf v} ds &= - \int_{\partial \Omega_1^{(1)}} pJ {{\bf F}^{-T}} {\bf N}_1^{(1)} \cdot {\bf v} ds, \tag{II}\\
	\int_{\partial \Omega_3^{(2)}} {\bf T N}_3^{(2)} \cdot {\bf v} ds &=0, \tag{III}\\
	\int_{\partial \Omega_1^{(2)}} {\bf T N}_1^{(2)} \cdot {\bf v} ds &+ \int_{\partial \Omega_2^{(1)}} {\bf T N}_2^{(1)} \cdot {\bf v} ds=0, \tag{IV}\\
	\int_{\partial \Omega_2^{(2)}} {\bf T N}_2^{(2)} \cdot {\bf v} ds &+ \int_{\partial \Omega_3^{(1)}} {\bf T N}_3^{(1)} \cdot {\bf v} ds=0, \tag{V}\\
	\int_{\partial \Omega_L} {\bf T N}_L \cdot {\bf v} ds &= 0, \tag{VI}\\
	\int_{\partial \Omega_R} {\bf T N}_R \cdot {\bf v} ds &= 0. \tag{VII}
	\end{align} 
	Adding equations (I)-(VII) and using $\Omega = \bigcup \limits_{i=1}^3 \Omega_i$  gives us
	
	\begin{align}
	&- \sum_{i=1}^3 \left[ \int_{\Omega_i} (\nabla \cdot {\bf T}_i)\cdot{\bf v} {\bf dx} \right]+ \int_{\Omega} {\bf f} \cdot {\bf v} {\bf dx} +  \int_{\partial \Omega_1^{(1)}} p J {\bf F}^{-T} {\bf N}_1^{(1)} \cdot {\bf v} {\bf dx}+ \int_{\partial \Omega_1^{(1)}} {\bf T N}_1^{(1)} \cdot {\bf v} ds \nonumber\\ 
	&+ \int_{\partial \Omega_2^{(1)}} {\bf T N}_2^{(1)} \cdot {\bf v} ds+ \int_{\partial \Omega_3^{(1)}} {\bf T N}_3^{(1)} \cdot {\bf v} ds + \int_{\partial \Omega_1^{(2)}} {\bf T N}_1^{(2)} \cdot {\bf v} ds +  \int_{\partial \Omega_2^{(2)}} {\bf T N}_2^{(2)} \cdot {\bf v} ds \nonumber \\
	&+ \int_{\partial \Omega_3^{(2)}} {\bf T N}_3^{(2)} \cdot {\bf v} ds + \int_{\partial \Omega_L} {\bf T N}_L \cdot {\bf v} ds + \int_{\partial \Omega_R} {\bf T N}_R \cdot {\bf v} ds = 0 \label{eq2.1.27.th1}
	\end{align}
	Now using the divergence theorem on the sum results in
	
	\begin{align}
	- \sum_{i=1}^3 \left[ \int_{\Omega_i} (\nabla \cdot {\bf T}_i)\cdot{\bf v} {\bf dx} \right]& =  \sum_{i=1}^3 \left[ \int_{\Omega_i}({\bf T}_i: \nabla {\bf v}) {\bf dx} \right] -\int_{\partial \Omega_1^{(1)}} {\bf T N}_1^{(1)} \cdot {\bf v} ds - \int_{\partial \Omega_2^{(1)}} {\bf T N}_2^{(1)} \cdot {\bf v} ds \nonumber\\
	& - \int_{\partial \Omega_3^{(1)}} {\bf T N}_3^{(1)} \cdot {\bf v} ds - \int_{\partial \Omega_1^{(2)}} {\bf T N}_1^{(2)} \cdot {\bf v} ds -  \int_{\partial \Omega_2^{(2)}} {\bf T N}_2^{(2)} \cdot {\bf v} ds \nonumber \\
	& - \int_{\partial \Omega_3^{(2)}} {\bf T N}_3^{(2)} \cdot {\bf v} ds - \int_{\partial \Omega_L} {\bf T N}_L \cdot {\bf v} ds - \int_{\partial \Omega_R} {\bf T N}_R \cdot {\bf v} ds \label{eq2.1.27.th2}
	\end{align}
	By replacing (\ref{eq2.1.27.th2}) in (\ref{eq2.1.27.th1}) we get
	\begin{align}
	\sum_{i=1}^3 \left[ \int_{\Omega_i}({\bf T}_i: \nabla {\bf v}) {\bf dx} \right] +  \int_{\Omega} {\bf f} \cdot {\bf v}  {\bf dx} +  \int_{\partial \Omega_1^{(1)}} p J {\bf F}^{-T} {\bf N} \cdot {\bf v} ds  = 0 \label{eq2.1.27.th3}
	\end{align}
	Switching to cylindrical coordinates we get
	\begin{align*}
	2\pi \sum_{i=1}^{3} \left[\int_{-L}^{L} \int_{A_i}^{A_{i+1}}\left[ ( {\bf T}_i: \nabla{\bf v})  + {\bf f} \cdot {\bf v}  \right]\ R \ dR  \ dZ  \right]+2\pi \left.\left(\int_{-L}^{L} p J {\bf F}^{-T} {\bf N} \cdot {\bf v} \ R \ dZ \right)\right|_{R=A_1} = 0 
	\end{align*}
	Notice that since there is no dependence on $\Theta$ due to axisymmetry we have integrated with respect to $\Theta$ producing the $2 \pi$ coefficients. Using the definition  ${\bf T}_i = J_{g_i} \pt{W_i}{{\bf F}_{e_i}} {\bf G}_i^{-T}$ we get (\ref{eq2.1.27}) for every  ${\bf v} \in \mathit{C}^{\infty}(\Omega)$. 
\end{proof}

{\bf Note}: In this paper we assume that the body forces are negligible. Therefore, (\ref{eq2.1.27}) turns into
\begin{align}
2\pi \sum_{i=1}^{3}\int_{-L}^{L} \int_{A_i}^{A_{i+1}} \left(J_{g_i} \pt{W_i}{{\bf F}_{e_i}} {\bf G}_i^{-T} : \nabla {\bf v}\right)  \ R \ dR \ dZ 
+ 2\pi \left.\left(p \int_{-L}^{L}  J {\bf F}^{-T} {\bf N} \cdot {\bf v} \ R \ dZ \right)\right|_{R=A_1}=0 \label{eq2.1.28}
\end{align}
for every ${\bf v} \in \mathit{C}^{\infty}(\Omega)$. We use the following table for parameter values.
\begin{center}
	\begin{tabular}{||c|c|c||} 
		\hline
		Symbol & Units & Value\\ [0.5ex] 
		\hline\hline
		$\mu_1$ &  kPa & $27.9$ \\ 
		\hline
		$\mu_2$ &  kPa & $1.27$\\
		\hline
		$\mu_3$ & kPa & $7.56$   \\
		\hline
		$\nu$ & Dimensionless & $0.49$ \\
		\hline
		$\eta_1$ & kPa & $263.66$\\
		\hline
		$\eta_2$ & kPa & $21.60$\\
		\hline
		$\eta_3$ & kPa & $38.57$\\
		\hline
		$\beta_1$ & Dimensionless & $170.88$\\
		\hline
		$\beta_2$ & Dimensionless & $8.21$\\
		\hline
		$\beta_3$ & Dimensionless & $85.03$\\
		\hline
		$\rho_1$ & Dimensionless & $0.51$\\
		\hline
		$\rho_2$ & Dimensionless & $0.25$\\
		\hline
		$\rho_3$ & Dimensionless & $0.55$\\
		\hline
		$\varphi_1$ & Degrees & $60.3$\\
		\hline
		$\varphi_2$  & Degrees & $20.61$\\
		\hline
		$\varphi_3$ & Degrees & $67$\\
		\hline
		$A_1$ & mm & $3$\\
		\hline
		$A_2$ & mm & $3.5$\\
		\hline
		$A_3$ & mm & $4.5$\\
		\hline
		$A_4$ & mm & $5.5$\\
		\hline
		$L$ & mm & $40$\\[1ex]
		\hline
	\end{tabular}
	\captionof{table}{\footnotesize{List of parameter values used in this paper. Mechanical parameters taken from \cite{Holzapfel2005}. The values of $A_k$, $k=1, \dots 4$ and $L$ are estimated.}}
\end{center}

Ultimately, we need to find a displacement field $\bf u$ that gives us a deformation gradient $\bf F$ in (\ref{eq2.1.20}) and (\ref{eq2.1.22}) which gives us the elastic tensors ${\bf F}_{e_i}$ for each layer by (\ref{eq2.1}) which leads to the first Piola-Kirchoff stress tensors ${\bf T}_i$ for each layer that gives (\ref{eq2.1.28}). 

We use FEniCS as our computing platform for solving this problem numerically. FEniCS is a powerful and open source package that can be utilized by languages such as C++ and Python (see \cite{Fenics2016}). For this problem we use a 2D mesh in $(R,Z)$ with about 11000 triangles. To avoid shear locking we use second order elements. This way we approximate the displacement field by second order Lagrangian elements which leads to a linear approximation for the strain. Also increasing the number of elements along the thickness of the domain is another common remedy for shear locking (see \cite{Zienkiewicz2005}). Although the problem is computationally intensive, the University of Delaware's Caviness cluster was able to find solutions in about 12 hours. Thanks to access to a high performance computing resource we were able to take advantage of both measures to simulate growth for large values of $t$ in (\ref{eq2.1.24})-(\ref{eq2.1.26}). 

\section{Results and Discussion}
For the rest of this paper we consider the blood pressure $p =12 \text{ kPa}= 90 \text{ mmHg}$ and we assume that the artery is in the pressurized state at $t=0$, see Figure \ref{fig2}. The blood pressure causes the radii $A_1$, $A_2$, $A_3$ and $A_4$ to increase and as a result the length of the artery decreases to conserve the volume. We use the notations $a_1$, $a_2$, $a_3$ and $a_4$ to refer to the radii in the deformed domain from now on. 

\begin{figure}[h]
	\centering
	\includegraphics[scale=0.84]{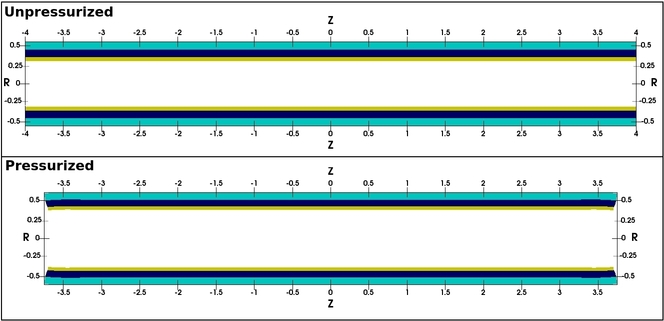}
	\caption{\footnotesize{\textit{Top}: The unpressurized reference domain. \textit{Bottom}: The reference domain after applying the blood pressure of $12$ kPa. } } \label{fig2}
\end{figure}

\subsection{The effect of pure growth in each direction}
First we start with investigating the effect of pure growth in each direction separately. We can roughly see in Figure \ref{fig3abc} the effect of such growth.  We believe that their different behaviors will give an insight on how each component of the growth tensor contributes to the overall process of remodeling. We provide graphs such as lumen area as a function of stenosis and lumen area as a function of time given the definition
\begin{equation}\label{eq3.1.29}
\text{Stenosis}(Z) = \frac{\text{Intima Area}(Z)}{\text{Intima Area}(Z)+\text{Lumen Area}(Z)}
\end{equation}
In addition, we explore the stress profiles in each direction as well as changes in the fiber angles and strain energy.

\begin{figure}[H]
	\begin{minipage}{\textwidth}
		\subfloat[]{\includegraphics[scale=0.6]{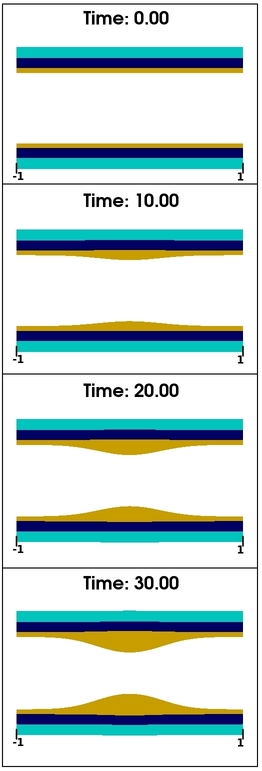}} 
		\hspace*{0.3in}\subfloat[]{\includegraphics[scale=0.6]{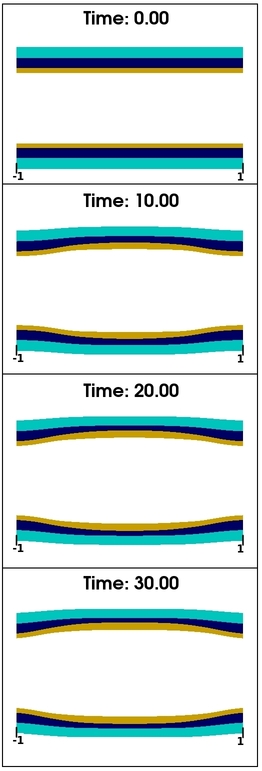}}
		\hspace*{0.3in}\subfloat[]{\includegraphics[scale=0.6]{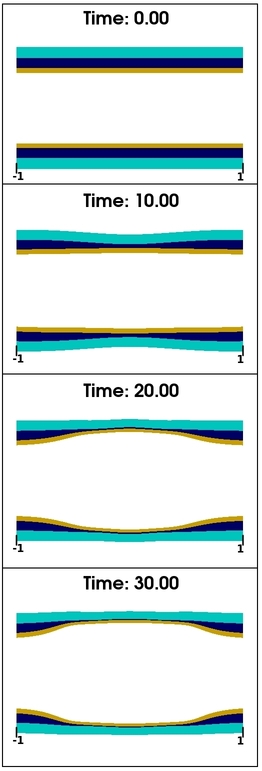}}
	\end{minipage}
	\caption{\footnotesize{Evolution of the domain for $-1\leq Z \leq1$ subject to (a) pure radial growth with $(\alpha,\beta,\gamma)=(1,0,0)$, (b) pure circumferential growth with $(\alpha,\beta,\gamma)=(0,1,0)$ and (c) pure axial growth with $(\alpha,\beta,\gamma)=(0,0,1)$. Parameter $a$ in (\ref{eq2.1.24})-(\ref{eq2.1.26}) is taken to be 6. } } \label{fig3abc}
\end{figure} 

\subsubsection{Pure Radial Growth} 

Let us assume that the intima grows according to the growth tensor ${\bf G}_{\alpha} = \textrm{diag}(g_{\alpha}(Z),1,1)$. This means that the intima grows radially by $g_{\alpha}(Z)$ from (\ref{eq2.1.24}) and there is no growth in the circumferential and radial direction. 

According to Figure \ref{fig3abc}(a) the radial growth almost exclusively contributes to inward thickening of the intima. As expected the inward remodeling is greater when closer to the center of growth $Z=0$ and consequently the artery undergoes more stenosis there, see Figure \ref{fig4}. We refrained from including more cross sections since far away from $Z=0$ the growth function has little to no effect and therefore the lumen area stays the same.

\begin{figure}[H]
	\centering
	\includegraphics[scale=0.5]{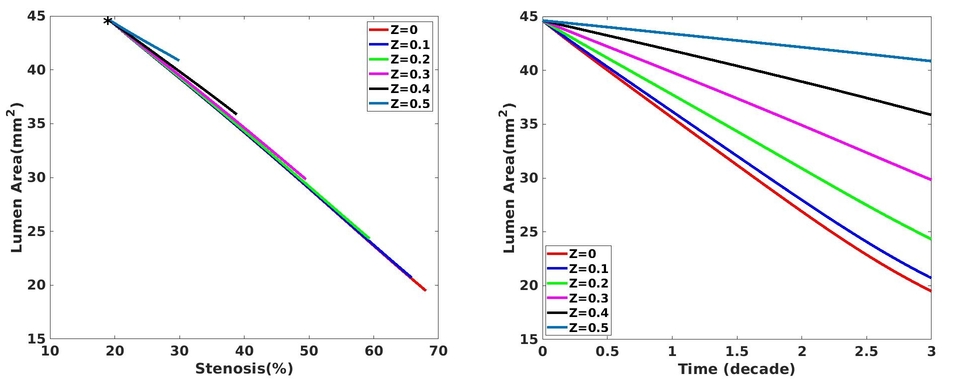}
	\caption{\footnotesize{\textit{Left}: Lumen area against stenosis. Star denotes the time $t=0$.  \textit{Right}: Lumen area in time. These graphs show that with pure radial growth close to the center of growth $Z=0$ the remodeling is strictly inward. }} \label{fig4}
\end{figure}\raggedbottom

Moreover, we can see that pure radial growth does not significantly change the maximum magnitude of stresses in the intima, see Figure \ref{fig5abc}(a). However, one can see a slight change in the maximum compressive stresses of the media and adventitia in Figures \ref{fig5abc}(b) and \ref{fig5abc}(c). We can conclude that when inward remodeling in the intima via radial growth reaches a certain level it affects the two other layers. In Figure \ref{fig3abc}(a) after $t=20$, media and adventitia experience a compressive force imposed by the intima. As a result they are slightly pushed back  which corresponds to the increases in $a_2$, $a_3$ and $a_4$ in Figure \ref{fig5abc}(e). On the other hand these changes in the radii are such that the media and adventitia thickness remain roughly the same. Also Figure \ref{fig5abc}(d) shows that the fiber angles only change in the intima by increasing. We can also see that pure radial growth does  does not greatly affect the underlying strain energy and thus stresses, see section 3.2.

\begin{figure}[H]
	\begin{minipage}{\textwidth}
		\subfloat[]{\includegraphics[scale=0.46]{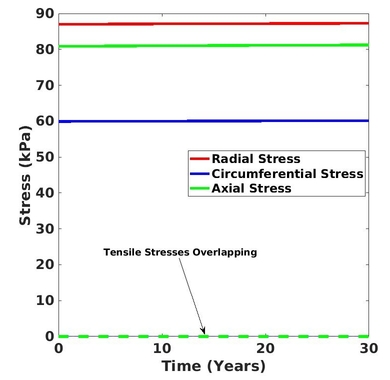}} 
		\subfloat[]{\includegraphics[scale=0.46]{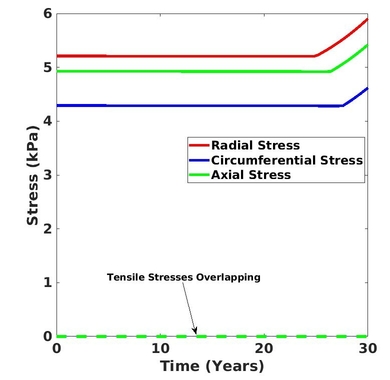}}
		\subfloat[]{\includegraphics[scale=0.46]{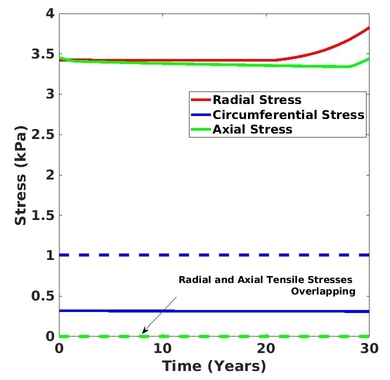}}\\
		\hspace*{0.85in}\subfloat[]{\includegraphics[scale=0.46]{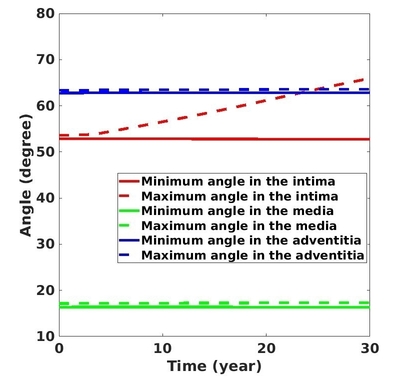}}
		\subfloat[]{\includegraphics[scale=0.46]{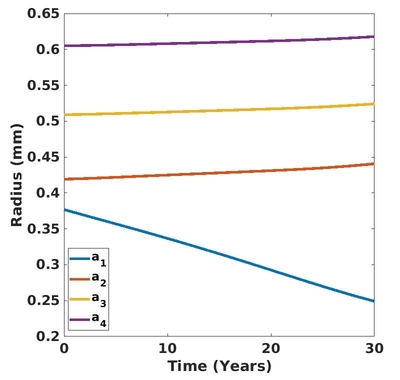}}
	\end{minipage}
	\caption{\footnotesize{The above graphs belong to the case of pure radial growth with $(\alpha,\beta,\gamma)=(1,0,0)$. {\it Top:} Changes in the maximum magnitude of compressive ({\it solid}) and tensile ({\it dashed}) stress components in the (a) intima (b) media and (c) adventitia. In (a) and (b) the stresses are mainly compressive since the tensile stresses are zero for all $t$. {\it Bottom:} (d) Changes in the minimum and maximum fiber angles. (e) Changes in the radii at $Z=0$.} } \label{fig5abc}
\end{figure}



\subsubsection{Pure Circumferential Growth}

Now we assume that growth is purely in the circumferential direction. Therefore the growth tensor takes the form ${\bf G}_{\beta} = \textrm{diag}(1,g_{\beta}(Z),1)$.

According to Figure \ref{fig3abc}(b), pure circumferential growth mostly contributes to the outward remodeling of the vessel. There is a slight intimal thickening and increase in stenosis but compared to the radial growth it is negligible, see Figure \ref{fig6}. Also one can see that the lumen area plateaus in Figure \ref{fig6} for large $t$ which might be due to the effect of stiffening collagen fibers. 

Unlike the radial growth, circumferential growth has a large compressive effect on each layer, see Figures \ref{fig7abcde}(a)-(c). We can see that the media thickness decreases significantly in Figure \ref{fig7abcde}(e). The abrupt increase in the magnitude of the maximum compressive stress components starts after $t=10$ in all three layers which is about the same time that media starts thinning. Also according to Figure \ref{fig7abcde}(d) the fiber angles in all three layers decrease. This is due to the outward remodeling imposed by the circumferential growth. Also this decrease in each layer happens quickly at the beginning and then becomes very slow which might be due to the fibers getting stiffer. As a result circumferential growth affects the stresses and the underlying strain energy more significantly than the radial growth, see section 3.2. 
\begin{figure}[H]
	\centering
	\includegraphics[scale=0.5]{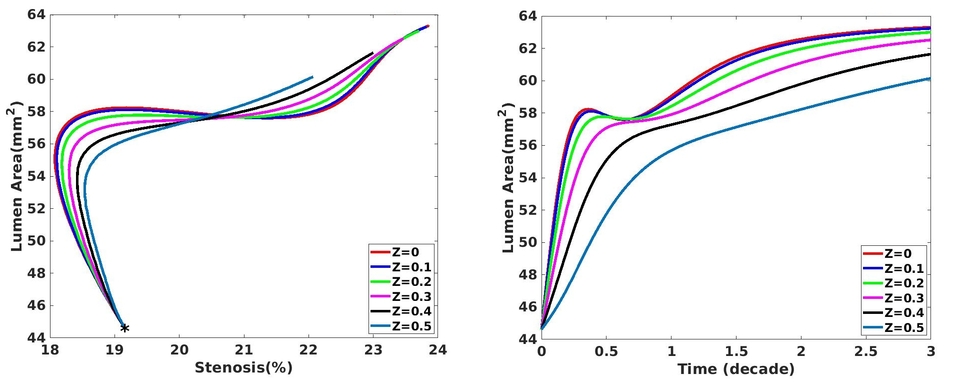}
	\caption{\footnotesize{\textit{Left}: Lumen area against stenosis. Star denotes the time $t=0$. \textit{Right}: Lumen area in time. These graphs show that with pure circumferential growth close to the center of growth $Z=0$ the remodeling is mostly outward.  } } \label{fig6}
\end{figure}\raggedbottom

\begin{figure}[H]
	\begin{minipage}{\textwidth}
		\subfloat[]{\includegraphics[scale=0.46]{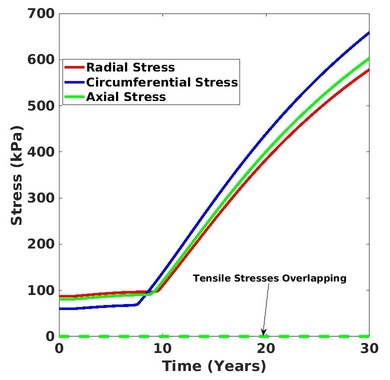}} 
		\subfloat[]{\includegraphics[scale=0.46]{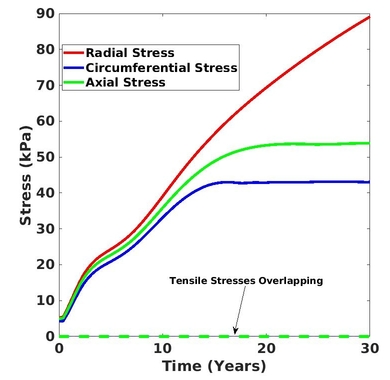}}
		\subfloat[]{\includegraphics[scale=0.46]{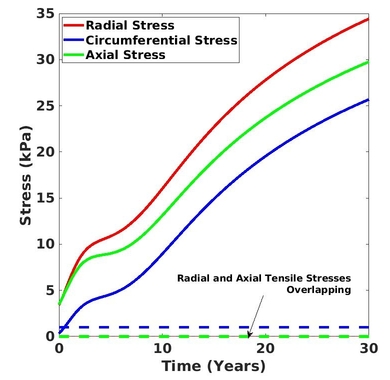}}\\
		\hspace*{0.9in}\subfloat[]{\includegraphics[scale=0.46]{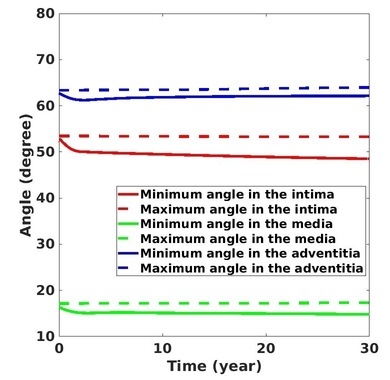}}
		\subfloat[]{\includegraphics[scale=0.46]{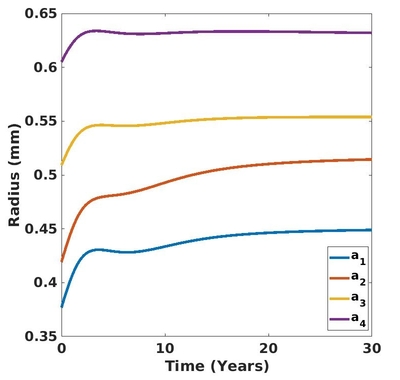}}
	\end{minipage}
	\caption{\footnotesize{The above graphs belong to the case of pure circumferential growth with $(\alpha,\beta,\gamma)=(0,1,0)$. {\it Top:} Changes in the maximum magnitude of compressive ({\it solid}) and tensile ({\it dashed}) stress components in the (a) intima (b) media and (c) adventitia. In (a) and (b) the stresses are mainly compressive since the tensile stresses are zero for all $t$. {\it Bottom:} (d) Changes in the minimum and maximum fiber angles. (e) Changes in the radii at $Z=0$.} } \label{fig7abcde}
\end{figure}


\subsubsection{Pure Axial Growth}

Finally we do the same investigation when growth is purely in the axial direction. Therefore we take the growth tensor to be ${\bf G}_{\gamma} = \textrm{diag}(1,1,g_{\gamma}(Z))$.

Pure axial growth mainly contributes to axial stretch, see Figure \ref{fig3abc}(c). The intimal thickening and stenosis are generally mild and as a matter of fact this mode of growth almost exclusively increases the lumen area and as a result {\it reduces} stenosis (Figure \ref{fig8}). We speculate that the cause for the lumen area plateauing in time is the stiffening effect of the collagen fibers. 

The axial growth has an interesting effect on the stress components. According to Figure \ref{fig9abcde} even without looking at the fiber angles or radii we can tell that something crucial happens after $t=10$. In the intima the effect is purely compressive (Figure \ref{fig9abcde}(a)) while in the media there is a large tensile stress in the circumferential and axial directions and a large compressive stress in the radial one (Figure \ref{fig9abcde}(b)). Axial growth also induces a compressive stress in the adventitia. Looking at Figure \ref{fig9abcde}(e) shows that there is a quick decrease in the thickness of all of three layers around $t=10$ and that is when the large increase in the stress components occurs. 

Looking at Figure \ref{fig1ab}(b), one can picture what happens if the vessel is stretched axially as a result of the axial growth. Since by definition fiber angles are measured with respect to $\hat{\bf e}_{\Theta}$ (see eq. (\ref{eq2.1.23})) stretching in the axial direction causes the fiber angles to become larger at first and then around $t=10$ they start to decrease, see Figure \ref{fig9abcde}(d). Given that the length of the growing region in the intima is nearly one fourth the length of the domain, elements far from $Z=0$ are basically unaffected by growth. The differential growth gives rise to an axial compression, thus decreasing the fiber angles. Again by looking at Figure \ref{fig9abcde} we can see that the axial growth causes a large increase in stresses in each layer after $t=10$. Perhaps the most interesting change happens inside of the media (\ref{fig9abcde}(b)). We can see a large increase in {\it tensile} circumferential and axial stresses which is because the vessel dilates quickly after $t=10$. Among all the other cases axial growth induces the biggest change in the neighboring layers. As one can see in Figure \ref{fig3abc}(c) the intima has to find a way to expand circumferentially and axially. Therefore it imposes tensile stresses on the media stretching it in those directions. On the other hand, media's resistance to these changes puts the intima under compressive stresses.   The axial growth also has the most significant effect on the underlying strain energy, see section 3.2.

\begin{figure}[H]
	\centering
	\includegraphics[scale=0.5]{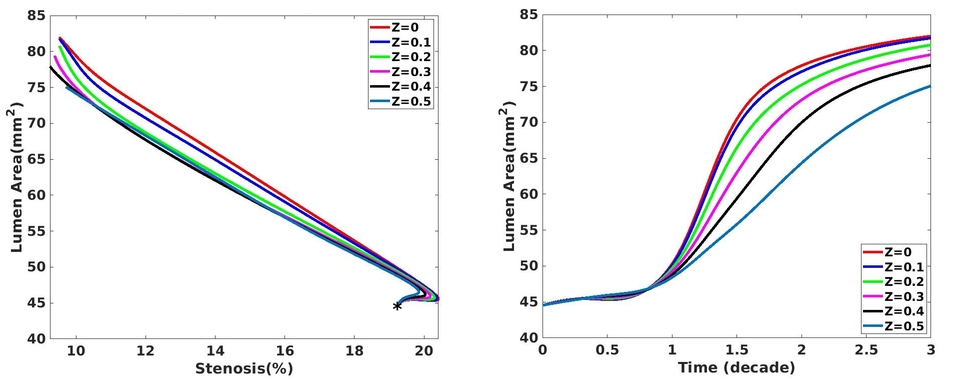}
	\caption{\footnotesize{\textit{Left}: Lumen area against stenosis. Star denotes the time $t=0$. \textit{Right}: Lumen area in time. These graphs show that with pure axial growth close to the center of growth $Z=0$ the remodeling is mostly outward with no significant increase in stenosis.} } \label{fig8}
\end{figure}

In conclusion, we saw that larger changes in fiber orientation occur when the growth is accompanied by outward remodeling. Also the angles increase as a result of axial stretch and decrease as a result of circumferential stretch. Among all three cases the radial growth induces the smallest change in the fiber orientation and stress components. We speculate that whenever the growth in the intima causes a reduction in the media and adventitia thickness the stresses are higher. This was mainly accompanied by outward remodeling. We saw that the axial growth had the most significant impact on the thickness of the layers and consequently caused a large change in the magnitude of stress components. Next we define a strain energy function and compare the energy change in each of the cases.

\begin{figure}[H]
	\begin{minipage}{\textwidth}
		\subfloat[]{\includegraphics[scale=0.46]{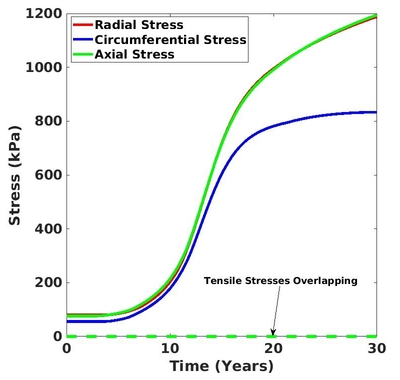}} 
		\subfloat[]{\includegraphics[scale=0.46]{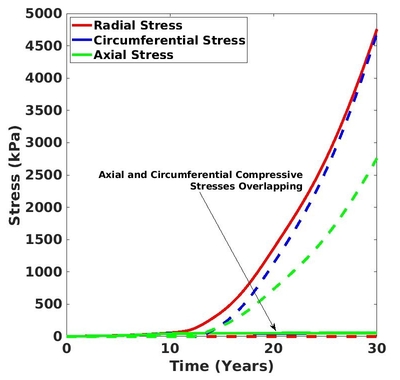}}
		\subfloat[]{\includegraphics[scale=0.46]{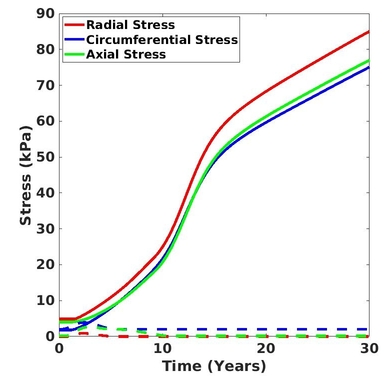}}\\
		\hspace*{0.9in}\subfloat[]{\includegraphics[scale=0.46]{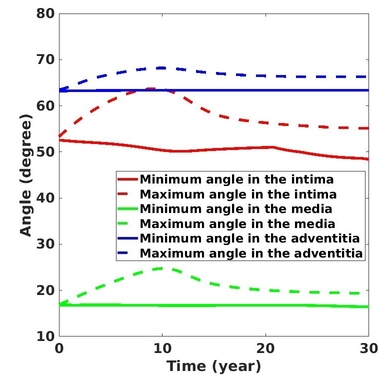}}
		\subfloat[]{\includegraphics[scale=0.46]{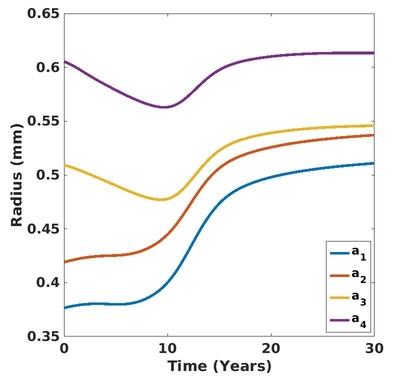}}
	\end{minipage}
	\caption{\footnotesize{The above graphs belong to the case of pure axial growth with $(\alpha,\beta,\gamma)=(0,0,1)$. {\it Top:} Changes in the maximum magnitude of compressive ({\it solid}) and tensile ({\it dashed}) stress components in the (a) intima (b) media and (c) adventitia. In (a) the stresses are mainly compressive since the tensile stresses are zero for all $t$. {\it Bottom:} (d) Changes in the minimum and maximum fiber angles. (e) Changes in the radii at $Z=0$.}  } \label{fig9abcde}
\end{figure}

%

\subsection{Growth and Energy Change}
In this section we are interested in calculating the energy stored in the artery due to each growth direction. We want to see which growth direction is more energetically favorable. 
The total energy consists of the bulk energy and the energy produced by the external forces such as the blood pressure. However, after enough time has passed the change in the bulk energy gets much larger than the one imposed by the blood pressure. Therefore for simplicity we only consider changes in the total bulk energy,
\begin{equation}
E = 2 \pi \sum_{i=1}^3 \int_{-L}^{L} \int_{A_i}^{A_{i+1}}J_{g_i} W_i \ R \ dR \ dZ. \label{eq3.1.30}
\end{equation}
Where $W_i$ is defined as $\Psi_{\rm vol}^i+\Psi_{\rm aniso}^i$ from (\ref{eq2.5}) and (\ref{eq2.6}). To show that (\ref{eq3.1.30}) is the right strain energy for our problem we prove:

\begin{theorem}
	The displacement field ${\bf u} \in \mathit{C}^2(\Omega)$ that satisfies the cylindrical weak equation 
	\begin{equation}
	2\pi \sum_{i=1}^{3}\int_{-L}^{L} \int_{A_i}^{A_{i+1}} \left(J_{g_i} \pt{W_i}{{\bf F}_{e_i}} {\bf G}_i^{-T} : \nabla {\bf v}\right)   \ R \ dR \ dZ = 0  \label{eq3.1.31}
	\end{equation}
	for every ${\bf v} \in \mathit{C}^{\infty}(\Omega)$ also makes (\ref{eq3.1.30}) stationary.
\end{theorem}
\begin{proof}
	We just need to show that for (\ref{eq3.1.30}), the Gateaux derivative of $E$ at ${\bf u}$ is zero. In other words we have to prove
	\begin{equation*}
	\frac{d}{d \epsilon}E[{\bf u}+ \epsilon {\bf v}]\left. \right|_{\epsilon=0} = 0
	\end{equation*}
	for all  ${\bf v} \in \mathit{C}^{\infty}(\Omega)$.
	
	Take  ${\bf v} \in \mathit{C}^{\infty}(\Omega)$. Then small perturbations in ${\bf u}$ give rise to small perturbations in ${\bf F}$ and consequently ${\bf F}_{e_i}$. 
	\begin{equation*}
	{\bf u} \rightarrow {\bf u}+ \epsilon {\bf v} \Longrightarrow {\bf F} \rightarrow {\bf F}+ \epsilon {\bf F'} \Longrightarrow {\bf F}_{e_i} \rightarrow {\bf F}+ \epsilon {\bf F}_{e_i}'
	\end{equation*}
	Where ${\bf F'} = \nabla {\bf v}$ and ${\bf F}_{e_i}' = {\bf F'} {\bf G}_i^{-1}$. Therefore
	\begin{align*}
	\frac{d}{d \epsilon}E[{\bf u}+ \epsilon {\bf v}]\left. \right|_{\epsilon=0} &= \frac{d}{d \epsilon} \left(2 \pi \sum_{i=1}^3 \int_{-L}^{L} \int_{A_i}^{A_{i+1}}J_{g_i} W_i({\bf F}+ \epsilon {\bf F}_{e_i}') \ R \ dR \ dZ \right)\\
	&= 2 \pi \sum_{i=1}^3 \int_{-L}^{L} \int_{A_i}^{A_{i+1}}J_{g_i} \frac{\partial W_i}{\partial {\bf F}_{e_i}} : {\bf F}_{e_i}' \ R \ dR \ dZ \\
	&=  2 \pi \sum_{i=1}^3 \int_{-L}^{L} \int_{A_i}^{A_{i+1}} \left(J_{g_i} \frac{\partial W_i}{\partial {\bf F}_{e_i}} {\bf G}_i^{-T} \right){\bf G}_i^{T} : {\bf F'}{\bf G}_i^{-1} \ R \ dR \ dZ \\
	& = 2 \pi \sum_{i=1}^3 \int_{-L}^{L} \int_{A_i}^{A_{i+1}} {\bf T}_i {\bf G}_i^{T} : {\bf F'}{\bf G}_i^{-1}  \ R \ dR \ dZ 
	\end{align*}
	Now using the tensor relations ${\bf A}^T : {\bf B}{\bf C} = {\bf C}{\bf A}: {\bf B}^T$ and ${\bf A}^T: {\bf B}^T = {\bf A}:{\bf B}$ we get
	\begin{align*}
	\frac{d}{d \epsilon}E[{\bf u}+ \epsilon {\bf v}]\left. \right|_{\epsilon=0} &= 2 \pi \sum_{i=1}^3 \int_{-L}^{L} \int_{A_i}^{A_{i+1}} {\bf G}_i^{-1} \left({\bf T}_i {\bf G}_i^{T}\right)^T : {\bf F'}  \ R \ dR \ dZ \\
	& = 2 \pi \sum_{i=1}^3 \int_{-L}^{L} \int_{A_i}^{A_{i+1}} {\bf G}_i^{-1} {\bf G}_i{\bf T}_i^T : {\bf F'}^T \ R \ dR \ dZ\\
	& = 2 \pi \sum_{i=1}^3 \int_{-L}^{L} \int_{A_i}^{A_{i+1}} {\bf T}_i : \nabla {\bf v} \ R \ dR \ dZ
	\end{align*}
	This is equal to zero by the assumption (\ref{eq3.1.31}). Thus, ${\bf u}$ is a stationary point for (\ref{eq3.1.30}).
\end{proof}
Now computing (\ref{eq3.1.30}) for the three cases we investigated in section 3.1.1, 3.1.2 and 3.1.3 we get Figure \ref{fig10}. We extract the energies induced by each growth direction at the times $t=5,10,15,20,25,30$ from Figure \ref{fig10} and we calculate their average values, see Table 2. Assuming that intima growth stems from cell division, the amount of energy needed for the cells to divide in the axial direction is on average 4 times bigger than that of the circumferential growth and on average 16 times bigger than that of the radial growth throughout the time (Table 2). In other words, growing in the axial direction is not energetically favorable while growing in the radial direction is more so. This observation motivates the need for an anisotropic treatment of the growth process.\\
\begin{center}
	\begin{tabular}{||c||c|c|c|c|c|c||c||} 
		\hline
		\backslashbox{Direction}{Time}&
		t=5 & t=10 & t=15 & t=20 & t=25 & t=30 & {Average} \\ [0.5ex] 
		\hline
		Radial & 3.251 & 3.255 & 3.257 & 3.262 & 3.288 & 3.356 & 3.277 \\
		\hline
		Circumferential & 3.435 & 4.270 & 7.411 & 13.059 & 22.097 & 32.807 & 13.846 \\
		\hline
		Axial & 3.851 & 8.427 & 32.156 & 63.211 & 91.828 & 117.195 & 52.778\\
		\hline
	\end{tabular}
	\captionof{table}{\footnotesize{Energy in $\mu J$ for 6 time values.} }
\end{center}
\begin{figure}[H]
	\centering
	\includegraphics[scale=0.54]{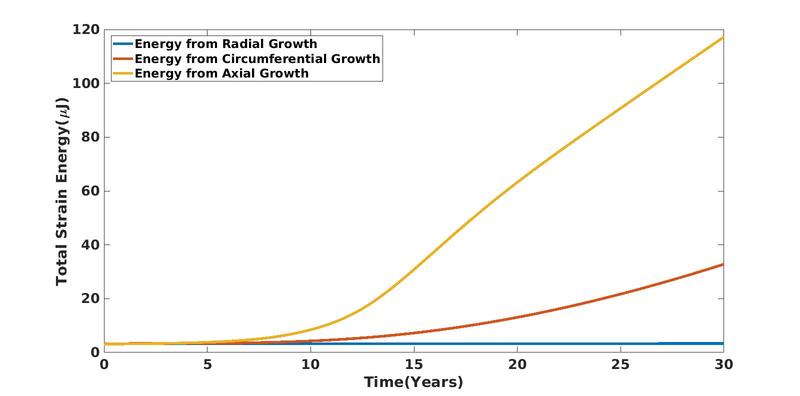}
	\caption{\footnotesize{The changes in total strain energy for growth in each direction.}} \label{fig10}
\end{figure}  

\subsubsection{Isotropic Growth vs General Anisotropic Growth}
Finally we can talk about general anisotropic growth. According to section 3.2 we choose $\alpha$, $\beta$ and $\gamma$ such that the energy change induced by the growth tensors ${\rm diag}(g_{\alpha},1,1)$, ${\rm diag}(1,g_{\beta},1)$ and ${\rm diag}(1,1,g_{\gamma})$ is roughly the same. As mentioned before the energy caused by the axial growth is on average 16 times and the energy caused by the circumferential growth is on average 4 times bigger than the energy caused by the radial growth for $0 \leq t \leq 30$. Therefore, we take $\alpha=1$, $\beta=\alpha/4=0.25$ and $\gamma=\alpha/16 \approx 0.06$ so that the energy change induced by growth in each direction is more balanced. We also consider an isotropic growth with $(\alpha,\beta,\gamma)=(0.25,0.25,0.25)$ to compare with the anisotropic growth. These parameter values are chosen so that both cases have approximately the same $J_g$ when $t \gg 1$. This means that the intima volume increase is roughly the same for both cases. We expect to see a more energetically favorable growth with the anisotropic case. 

\begin{figure}[H]
	\begin{minipage}{\textwidth}
		\hspace*{0.25in}\subfloat[]{\includegraphics[scale=0.66]{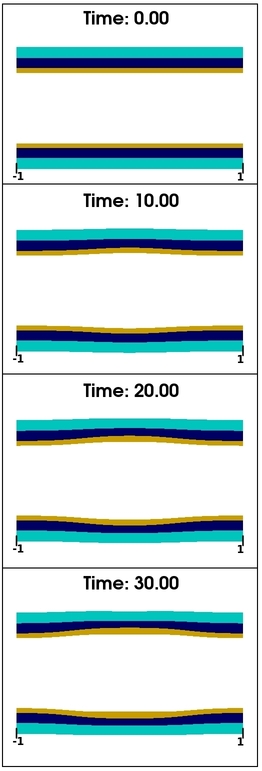}} \hspace*{0.95in}
		\subfloat[]{\includegraphics[scale=0.66]{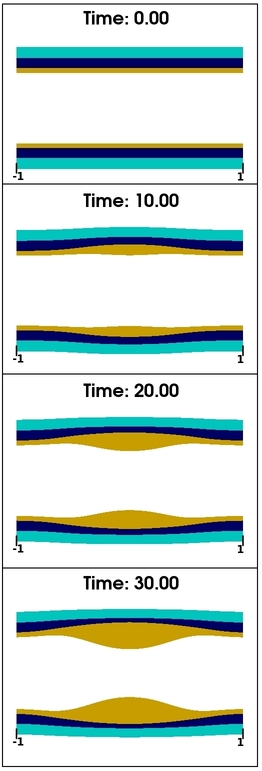}}
		
	\end{minipage}
	\caption{\footnotesize{Evolution of the domain for $-1 \leq Z \leq 1$ subject to (a) isotropic growth with $(\alpha,\beta,\gamma)=(0.25,0.25,0.25)$, (b) general anisotropic growth with $(\alpha,\beta,\gamma)=(1,0.25,0.06)$. Parameter $a$ in (\ref{eq2.1.24})-(\ref{eq2.1.26}) is taken to be 6. } } \label{fig11ab}
\end{figure} 
As usual we start by looking at the evolution of the domain in time, Figure \ref{fig11ab}. Due to the choice of small $\alpha$,$\beta$ and $\gamma$ for the isotropic case we do not observe a lot of changes in the domain. On the other hand, with approximately the same $J_g$ we can clearly observe an intimal thickening close to the center of growth accompanied by a significant encroachment for the anisotropic growth. Figure \ref{fig12} shows that the isotropic growth induces a sudden and unrealistic outward remodeling (with respect to stenosis) while the outward remodeling for the anisotropic growth (with respect to stenosis) is much gentler. In addition, the isotropic growth barely produces any inward remodeling during $0 \leq t \leq 30$ while the anisotropic growth results in $70\%$ stenosis at $Z=0$ due to the dominance of the radial growth direction, see Figure \ref{fig12}. This result is qualitatively closer to Glagov's original data which had a gentler compensation phase with respect to stenosis followed by an inward remodeling.

Looking at the stress profiles for the isotropic growth (Figure \ref{fig13abcde}(a)-(c)) we can see that the stresses are mainly compressive and very large in the intima, possibly the result of thinning of the intima, see Figure \ref{fig13abcde}(e). Although the stresses for the anisotropic growth are  also mostly compressive and most significant in the intima (see Figure \ref{fig14abcde}(a)-(c)) they are smaller and more physically reasonable (plaque caps are thought to rupture at about 300-545 kPa, see \cite{Kelly2013}). Notice that the maximum compressive stresses do not change for a while in Figure \ref{fig14abcde}(a) and then there is a sudden increase around $t=20$. We know from section 3.1.1 that the increase in the thickness and fiber angles is associated with an inward remodeling and does not significantly change the stresses and the strain energy. On the other hand, we saw in section 3.1.2 that the decrease in the thickness and fiber angles is associated with outward remodeling which changes the stresses and the strain energy more noticeably. Figures \ref{fig14abcde}(d) and (e) show that the anisotropic growth is a medley of these effects. The increase in the thickness and fiber angle in the intima seems more significant than the decreases because we considered the radial direction to be the dominant direction for growth. Also we do not see any changes in the fiber angles that resemble that of the axial growth due to suppression of the axial growth by the choice of small $\gamma$.  Therefore, the stresses in Figure \ref{fig14abcde}(a) (first following the radial growth trend) do not change for a while and then (following the circumferential growth trend after it kicks in) there is a sudden increase due to the circumferential growth effect. 

Figure \ref{fig13abcde}(d) shows the changes in the fiber angles for the isotropic growth. Because we have chosen a small value of $0.25$ for $\alpha$,$\beta$ and $\gamma$ the fiber angles do not change as significantly. However, we can see an increase in all three layers which is a characteristic of the axial growth. This much axial growth is enough to induce large stresses in the intima, see Figure \ref{fig13abcde}(a). Finally a comparison between the energies from the isotropic and anisotropic growth shows that the latter is much more energetically favorable, see Figure \ref{fig15}.

\begin{figure}[H]
	\centering
	\includegraphics[scale=0.48]{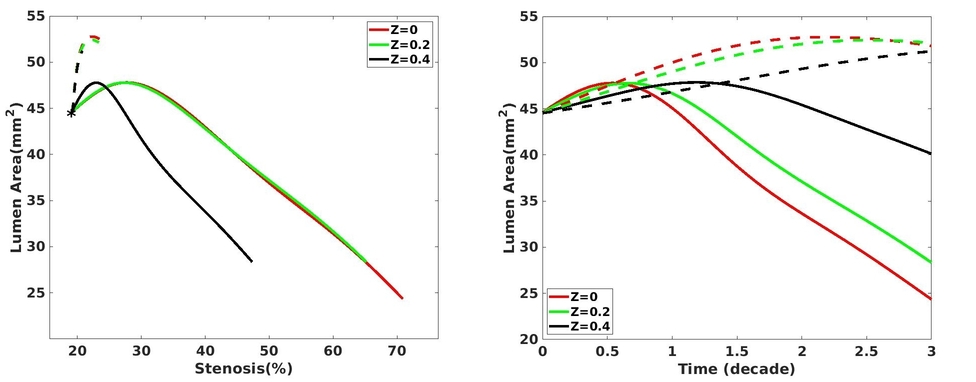}
	\caption{\footnotesize{\textit{Left}: Comparison between the Lumen area against stenosis for both isotropic and anisotropic growth. Star denotes the time $t=0$. \textit{Right}: Comparison between the lumen area in time for both isotropic and anisotropic growth. \textit{Dashed} lines correspond to isotropic and \textit{solid} lines correspond to anisotropic growth.} } \label{fig12}
\end{figure}\raggedbottom

\begin{figure}[H]
	\begin{minipage}{\textwidth}
		\subfloat[]{\includegraphics[scale=0.46]{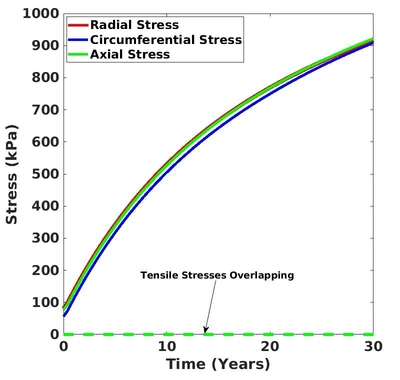}} 
		\subfloat[]{\includegraphics[scale=0.46]{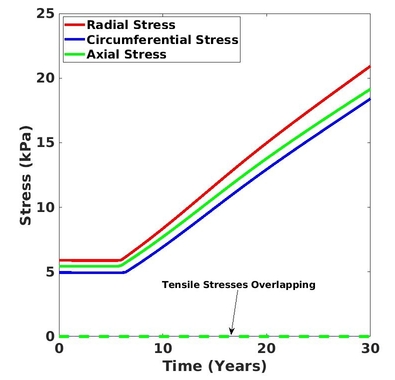}}
		\subfloat[]{\includegraphics[scale=0.46]{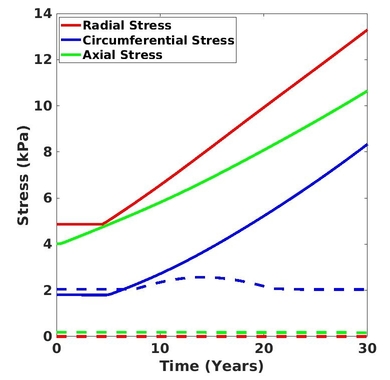}}\\
		\hspace*{0.9in}\subfloat[]{\includegraphics[scale=0.46]{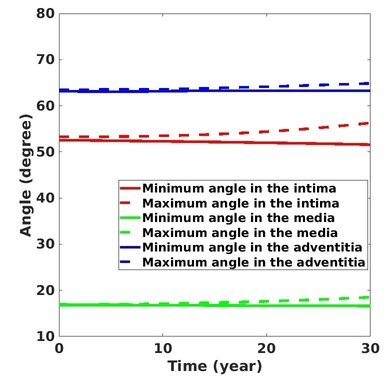}}
	    \subfloat[]{\includegraphics[scale=0.46]{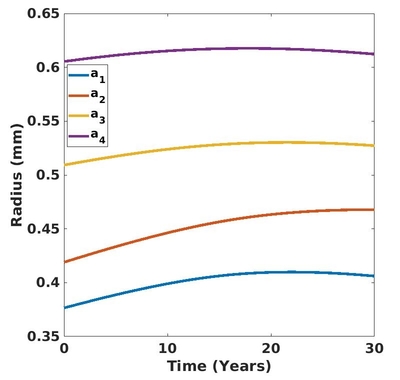}} 
	\end{minipage}
	\caption{\footnotesize{The above graphs belong to the case of isotropic growth with $(\alpha,\beta,\gamma)=(0.25,0.25,0.25)$. {\it Top:} Changes in the maximum magnitude of compressive ({\it solid}) and tensile ({\it dashed}) stress components in the (a) intima (b) media and (c) adventitia. In (a) and (b) the stresses are mainly compressive since the tensile stresses are zero for all $t$. {\it Bottom:} (d) Changes in the minimum and maximum fiber angles. (e) Changes in the radii at $Z=0$.}  } \label{fig13abcde}
\end{figure}

\begin{figure}[H]
	\begin{minipage}{\textwidth}
		\subfloat[]{\includegraphics[scale=0.46]{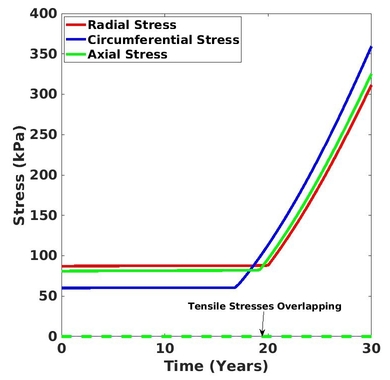}} 
		\subfloat[]{\includegraphics[scale=0.46]{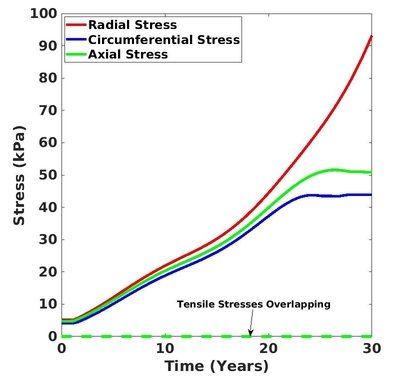}}
		\subfloat[]{\includegraphics[scale=0.46]{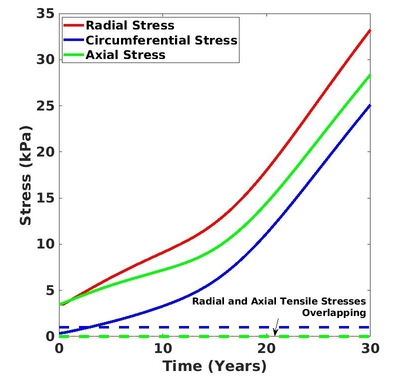}}\\
		\hspace*{0.9in}\subfloat[]{\includegraphics[scale=0.46]{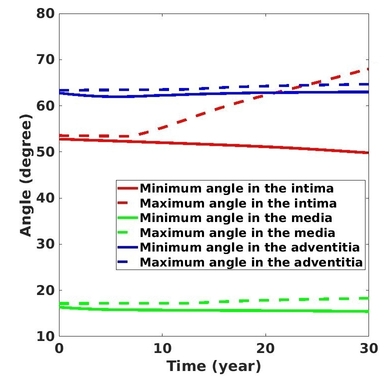}}
		\subfloat[]{\includegraphics[scale=0.46]{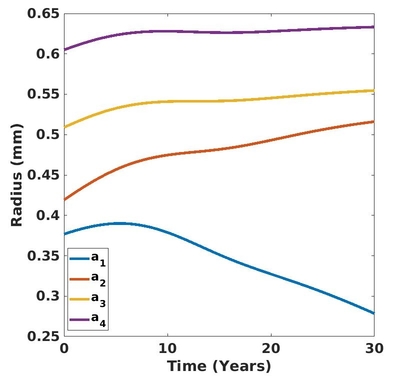}} 
	\end{minipage}
	\caption{\footnotesize{The above graphs belong to the case of anisotropic growth with $(\alpha,\beta,\gamma)=(1,0.25,0.06)$. {\it Top:} Changes in the maximum magnitude of compressive ({\it solid}) and tensile ({\it dashed}) stress components in the (a) intima (b) media and (c) adventitia. In (a) and (b) the stresses are mainly compressive since the tensile stresses are zero for all $t$. {\it Bottom:} (d) Changes in the minimum and maximum fiber angles. (e) Changes in the radii at $Z=0$.}  } \label{fig14abcde}
\end{figure}

\begin{figure}[H]
	\centering
	\includegraphics[scale=0.54]{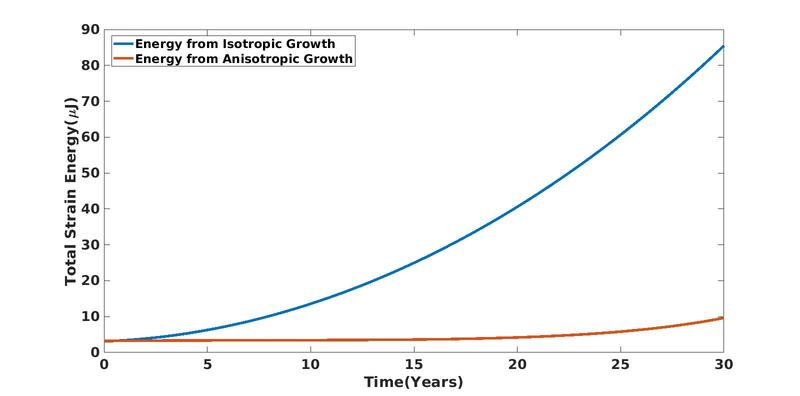}
	\caption{\footnotesize{Changes in the strain energy produced by isotropic and anisotropic growth in time.}} \label{fig15}
\end{figure}

\section{Conclusion}

In this paper we investigated an axisymmetric growing artery using morphoelasticity. First we were curious to see the effect of growth in each direction separately. This allowed us to associate each aspect of the Glagov remodeling phenomenon with growth in a certain direction. We observed that the radial growth is the main culprit in encroachment of the vessel while circumferential and axial growth are responsible for the outward remodeling. We also saw that radial growth thickens the intima while not changing the thickness of the other two layers. The circumferential growth thins the media without changing the thickness of the other two layers and the axial growth thins all the layers (media thins more significantly) and stretches the artery axially. Among all three growth directions the axial growth puts the artery under more stress than the other two cases while the radial growth has the least effect of such. Each growth has a different effect on the fiber angles. The radial growth causes the fiber angles in the intima to increase without changing the angles in the media and adventitia. The circumferential growth decreases the fiber angles in all three layers and the axial growth causes an increase and then a decrease in the fiber angles in all the layers. 

To enforce the need for an anisotropic treatment of our problem we first investigated the isotropic growth with the growth tensor ${\rm diag}(1+  0.25 t \exp(-aZ^2), 1+ 0.25 t \exp(-aZ^2)), 1+ 0.25 t \exp(-aZ^2))$. We saw mostly outward remodeling followed by a small amount of inward remodeling which was not realistic. Moreover, the magnitude of stresses were unphysical and the strain energy produced in this case was extremely large. Comparing the energy produced by each growth direction we deduced that the radial direction should be the dominant direction of growth since it is the most energetically favorable. Similarly, we decided to suppress the axial growth the most since it is the least energetically favorable. Therefore, we chose an anisotropic growth tensor of the form ${\rm diag}(1+  t \exp(-aZ^2),1+ 0.25 t \exp(-aZ^2)),1+ 0.06 t \exp(-aZ^2))$ which has approximately the same growth Jacobian $J_g$ as the isotropic case . On a plot of lumen area vs stenosis, we saw a much gentler outward remodeling followed by an encroachment starting later than that of the isotropic growth and a bigger stenosis. This was expected due to the dominance of growth in the radial direction. Also the stress profiles were more realistic and the strain energy produced was much less than the isotropic case.

Many studies have shown that smooth muscle cell (SMC) proliferation is the reason for intimal thickening (see \cite{Groves1995}, \cite{Sho2002}, \cite{Francis2003} and \cite{Nakagawa2018}). In addition, there are studies that suggest that SMC proliferation is regulated by changes in the circumferential stress in vessel wall (see \cite{Wayman2008}) or changes in the blood flow shear stress (see \cite{Ueba1997} and \cite{Haga2003}) . Although we think atherosclerosis proceeds because of growth in the intima, very few authors (maybe none) think about the nature of the anisotropy in the growth tensor. In other words, they do not determine which direction or directions this growth tends to progress. We suggested a way to choose an optimal anisotropic growth based on minimizing the strain energy. This choice was motivated purely by energy considerations and we do not know if there is any biological evidence to support this form of anisotropy. We believe if such biological reasons exist finding a way to control or change them might be the key to slow down or reverse the inward remodeling process which is responsible for affecting local hemodynamics and medical conditions such as angina.

\textbf{Acknowledgements}: NM was supported by a DE-CTR SHoRe pilot grant NIGMS IDeA U54-GM104941. PWF is supported by a Simons Foundation collaboration grant, award
number 282579.

 \vspace*{6pt}
\nocite{*}
\bibliographystyle{agsm}
\bibliography{bib4Glagov3D}
\end{document}